\documentclass[%
 reprint,
 amsmath,amssymb,
 aps,
 longbibliography,
 pra,
]{revtex4-2}
\usepackage{graphicx}
\usepackage{dcolumn}
\usepackage{bm}

\usepackage{float}
\usepackage{array, makecell} %
\usepackage{xcolor,colortbl}
\usepackage{framed,color}
\usepackage{xcolor}

\usepackage{amsthm,amsmath,amssymb}
\usepackage[bookmarks=true,hypertexnames=false,pagebackref]{hyperref}
\hypersetup{colorlinks=true, citecolor=blue, linkcolor=red,
  urlcolor=blue}
\usepackage{newpxtext}
\usepackage[libertine,vvarbb]{newtxmath}

\usepackage[T1]{fontenc} 
\usepackage{textcomp} 

\usepackage[scr=rsfso]{mathalfa}
\usepackage{nicefrac}
\usepackage{cleveref}
\usepackage{graphicx}
\usepackage[ruled,vlined]{algorithm2e}
\usepackage{aliascnt}
\usepackage[section]{placeins}
\usepackage{tcolorbox}
\usepackage{epigraph}
\usepackage{algpseudocode}
\usepackage{xifthen}
\usepackage{latexsym}
\usepackage{framed}
\usepackage{fullpage}
\usepackage{bm}
\usepackage{physics}
\usepackage{qcircuit}
\usepackage{tikz}
\usetikzlibrary{positioning}
\usepackage{caption}
\usepackage{subcaption}
\usepackage{multirow}
\usepackage{hyperref}
\usepackage{flushend}

\newtheorem{theorem}{Theorem}[section]

\newtheorem{lemma}[theorem]{Lemma}
\newtheorem{definition}[theorem]{Definition}

\SetKwInOut{Oracle}{Oracle}
\SetKwInOut{Parameter}{parameter}
\ResetInOut{output}

\newcommand{\sdotfill}{\textcolor[rgb]{0.8,0.8,0.8}{\dotfill}} 

\begin{document}

\title{Efficient Quantum Circuits for Machine Learning Activation Functions including Constant T-depth ReLU}

\author{Wei Zi$^{1,2}$}
\author{Siyi Wang$^{3}$}
\author{Hyunji Kim$^{4}$}
\author{Xiaoming Sun$^{1,2,}\! \!$}
\email{sunxiaoming@ict.ac.cn}
\author{Anupam Chattopadhyay$^{3,}\! \!$}
\email{anupam@ntu.edu.sg}
\author{Patrick Rebentrost$^{5,6,}$}
\email{cqtfpr@nus.edu.sg}
\affiliation{$^{1}$State Key Lab of Processors, Institute of Computing Technology, Chinese Academy of Sciences, Beijing, China}
\affiliation{$^{2}$School of Computer Science and Technology, University of Chinese Academy of Sciences, Beijing, China}
\affiliation{$^{3}$School of Computer Science and Engineering, Nanyang Technological University, Singapore}
\affiliation{$^{4}$Hansung University, Seoul, South Korea}
\affiliation{$^{5}$Centre for Quantum Technologies, National University of Singapore, Singapore}
\affiliation{$^{6}$School of Computing, National University of Singapore, Singapore}

\begin{abstract}
In recent years, Quantum Machine Learning (QML) has increasingly captured the interest of researchers. Among the components in this domain, activation functions hold a fundamental and indispensable role. Our research focuses on the development of activation functions quantum circuits for integration into fault-tolerant quantum computing architectures, with an emphasis on minimizing $T$-depth.
Specifically, we present novel implementations of ReLU and leaky ReLU activation functions, achieving constant $T$-depths of 4 and 8, respectively. 
Leveraging quantum lookup tables, we extend our exploration to other activation functions such as the sigmoid. This approach enables us to customize precision and $T$-depth by adjusting the number of qubits, making our results more adaptable to various application scenarios. 
This study represents a significant advancement towards enhancing the practicality and application of quantum machine learning. 
\end{abstract}

\maketitle


\section{Introduction}
Currently, Machine Learning (ML) is attracting a lot of attention, driven by its widespread adoption across various industries. Recent breakthroughs in natural language processing models, exemplified by GPT4~\cite{achiam2023gpt}, as well as advancements in image classification and generation \cite{affonso2017deep,nie2018medical,yi2018sharpness,li2019deep}, highlight the rapid advancements of this field. An important building block of machine learning architectures is feed-forward neural networks \cite{gurney2018introduction}, which are collections of layers of activation functions and linear transformations. 
In previous research, many different activation functions have been proposed, such as Sigmoid, Softmax, and ReLU, which are successfully utilized in ML applications.

Quantum machine learning~\cite{schuld2015introduction} has shown progress in both adapting common machine learning algorithms to run on future quantum computers and using machine learning techniques to learn about quantum systems \cite{peral2024systematic}. Some examples include Quantum Neural Networks (QNN) \cite{abbas2021power}, quantum support vector machines (QSVM) \cite{rebentrost2014quantum}, quantum principal component analysis (QPCA) \cite{lloyd2014quantum}, variational quantum eigensolver (VQE) \cite{tilly2022variational}, parameterized quantum circuits (PQC) \cite{benedetti2019parameterized}, and variations of the quantum approximate optimization algorithm (QAOA) \cite{zhou2020quantum}.

\begin{table*}
\centering
\caption{The $T$-depth of the quantum circuit implementing ReLU in both scenarios: with and without the constraint of qubit connectivity limited to a 2D grid. The $T$-depth for Leaky ReLU is also presented.}
\bgroup
\renewcommand{\arraystretch}{1.3}
\label{table:relu}
\centering
\resizebox{0.6\textwidth}{!}{
\begin{tabular}{c|c|c|c}
\hline
Activation function & ReLU & ReLU (on 2D grid) & Leaky ReLU \\ \hline
$T$-depth & 4 & 4 & 8 \\ \hline

\end{tabular}
}
\egroup
\end{table*}

Recently, numerous research findings on QNNs based on quantum circuit implementations have been proposed \cite{schuld2014quest,wan2017quantum,zhao2019building,tacchino2019artificial,cong2019quantum}. It turns out that the implementation of activation functions, using quantum circuits, is fundamental and crucial in QNNs. Various methods have been utilized to implement activation functions, including quantum phase estimation~\cite{schuld2015simulating,de2019implementing}, Taylor series expansion \cite{maronese2022quantum}, quantum oracle synthesis \cite{yan2020nonlinear}, and polynomial networks \cite{livni2014computational}. However, to the best of the authors' knowledge, there are no dedicated circuit designs for implementing activation functions with a focus on fault tolerance so far. In this work, we address this gap by constructing a quantum circuit using Clifford+$T$ gates. We specifically focus on minimizing the $T$-depth of the circuits, considering the high cost associated with fault-tolerant implementations of the $T$ gate and the limitation imposed by the coherence time of the quantum device \cite{gottesman2010introduction}. In particular, for some activation functions that have a high implementation complexity, such as the ReLU \cite{sharma2017activation}, approximating them using polynomials would result in high polynomial degrees \cite{goel2019learning}. 
We have specifically designed quantum circuits to implement them, significantly reducing the $T$-depth.

In particular, we propose a quantum circuit to implement the ReLU function with a constant $T$-depth of $4$, without using ancillary qubits. If the qubit connectivity is constrained to a 2D grid, the $T$-depth of the circuit for the ReLU function remains unchanged. Additionally, we implement Leaky ReLU functions using a quantum circuit with a constant $T$-depth of $8$. The results of ReLU and Leaky ReLU are summarized in Table~\ref{table:relu}.

For other activation functions, such as Sigmoid, SoftMax, Tanh, Swish, Exponential Linear Unit (ELU) \cite{sharma2017activation}, and Gaussian Error Linear Unit (GELU) \cite{hendrycks2016gaussian}, we utilize the quantum Lookup table (QLUT) \cite{krishnakumar2022aq} to implement them. In our implementation, we considered the trade-off between the number of qubits and implementation accuracy, as well as the trade-off between $T$-depth and ancilla count. Specifically, we represent the input and output of these activation functions using well-known floating-point numbers with $n \in \{8,16,32,64,128\}$ bits. The QLUT-based implementation method ensures that the representation accuracy of the floating-point numbers solely determines the implementation error. Furthermore, QLUT allows us to reduce the $T$-depth of the circuit by increasing the ancilla count. The detailed discussion is presented in Section.~\ref{sec:QLUT}.
The Qiskit implementation of quantum circuits for activation functions discussed in this paper is open-sourced on GitHub~\footnote{\url{https://github.com/ziwei1998/quantum-circuit-for-activation-functions}}.

The rest of the paper is organized as follows. Section \ref{sec:pre} presents preliminaries such as activation functions. Sections \ref{sec:relu} and 
\ref{sec:leakyrelu} focus on the quantum circuit implementations of ReLU and Leaky ReLU, respectively. Section \ref{sec:QLUT} explains how to utilize QLUT to implement other activation functions such as Sigmoid function. Finally, Section \ref{sec:concl} provides a conclusive summary of the entire paper.

\section{Preliminary}
\label{sec:pre}

Activation functions are integral components in machine learning and deep learning, providing the crucial non-linearity required for effective learning. In the classical world, several activation functions are commonly used including ReLU, Leaky ReLU, Sigmoid, SoftMax, Tanh, Swish, Exponential Linear Unit (ELU), Gaussian Linear Unit (GELU)~\cite{sharma2017activation,hendrycks2016gaussian}.
Specifically, The formulas for ReLU, Leaky ReLU, and sigmoid are as follows.

\begin{itemize}
    \item \textbf{ReLU:} $f(x) = \max(x,0)$.
    \item \textbf{Leaky ReLU ($0 < \alpha < 1$):}
    
    $
    f(x) = \max(x,\alpha x) =  \left\{
    \begin{array}{cc}
        x & x \geq 0 \\
        \alpha x & x < 0
    \end{array}
    \right.
    $.
    \item \textbf{Sigmoid:} ${f(x)} = 1/(1+e^{-x})$.
\end{itemize}

For a quantum circuit $U_f$ that can implement function $f$, it is defined as:
\begin{equation*}
U_f\ket{x} \ket{0}_{o} = \ket{x} \ket{f(x)}.
\end{equation*}

Here, $\ket{x}$ represents the quantum state that encodes the input of function $f$, $\ket{0}_o$ represents the output qubits used to store the result $f(x)$, and $\ket{f(x)}$ represents the quantum state that encodes the output $f(x)$. The main objective of this paper is to construct $U_f$ for various activation functions using Clifford gates and $T$ gates. Specifically, we utilize the quantum gates from the set $\{CNOT, H, S, T\}$ to build the circuits. The commonly used $X$ gate can be decomposed into four Clifford gates: $X = HSSH$. Due to the higher fault-tolerant implementation cost of the $T$ gate compared to Clifford gates, our focus lies in minimizing the $T$-depth, which represents the circuit's runtime \cite{gottesman2010introduction}. 
In our construction, three types of ancillary qubits may be utilized. Clean ancillary qubits are qubits initialized to $\ket{0}$ and reset to $\ket{0}$ after their use. Dirty ancillary qubits do not require the initial state but need to be reset after their use. Garbage qubits are initialized to $\ket{0}$ and do not need to be reset after use.

\section{Quantum Circuit Implementation of ReLU Function}
\label{sec:relu}

\begin{definition}
    A circuit $C_n$ is capable of implementing the $n$-bit ReLU function, is defined as:
    \begin{equation*}
        C_n\ket{\boldsymbol{x}}\ket{0}_{n-1} = \ket{\boldsymbol{x}}\ket{\max(\boldsymbol{x},0)},
    \end{equation*}
    where $\boldsymbol{x} \in \{0,1\}^n$ represents a fixed-point number, and the circuit $C_n$ is composed of Clifford+$T$ gates.
\end{definition}

Note that the input $\boldsymbol{x} \in \{0,1 \}^n$ represents a fixed-point integer, where the most significant bit (MSB) corresponds to the sign bit. Consequently, the output $\max(\boldsymbol{x},0)$ requires only $n-1$ bits for storage. To achieve this, we utilize $n-1$ output qubits initialized to $\ket{0}$. Thus, a total of $2n-1$ qubits are required to implement the ReLU function.

To understand why this number is necessary, consider $2^{n-1} + 1$ non-positive inputs. For all of these inputs, the output is $0$. Storing this output value of $0$ necessitates $n-1$ qubits. Moreover, to maintain reversibility in the presence of $2^{n-1} + 1$ distinct inputs yielding the same output $0$, an additional $n$ qubits are needed. Consequently, $2n-1$ qubits are both necessary and sufficient for this purpose.

To implement the ReLU function, it suffices to copy the remaining $n-1$ input qubits to the output qubits when the sign bit is $0$. This statement remains valid irrespective of whether we represent the input $\boldsymbol{x}$ using the 2's complement, the 1's complement, or the true form. Subsequently, we will demonstrate the implementation of the ReLU function using two architectures. The first architecture permits the application of the CNOT gate to any two qubits, while the second architecture restricts the application of the CNOT gate to qubits that are adjacent on a 2D grid.

\subsection{Implementation of ReLU Function without Connectivity Constraints}

The quantum fan-out gate $F_n$ is defined as $F_n\ket{a,b_1,\dots,b_n} = \ket{a,b_1 \oplus a, \dots, b_n \oplus a}$. Its well-known $O(\log n)$ depth implementation has been widely utilized \cite{takahashi2021power}, therefore we omit the proof of the lemma~\ref{lem:fanout}. For illustration, the implementation of $F_8$ is shown in Fig.~\ref{fig:F_8}.

\begin{lemma}
\label{lem:fanout}
    The quantum fan-out gate $F_n$ can be implemented by a quantum circuit with depth $O(\log n)$ and size $O(n)$ without ancillary qubits, using only CNOT gates.
\end{lemma}

\begin{figure}[h]
\centering
\begin{flushleft}\ \ \ \ \ \ \ \ \ \Qcircuit @C=0.85em @R=1em @!R{
    \lstick{\ket{a}} & \ctrl{1} & \qw &&& \qw & \qw & \qw & \ctrl{1} & \qw & \qw & \qw & \qw & \rstick{\ket{a}} \\
    \lstick{\ket{b_1}} & \targ \qwx[1] & \qw &&& \ctrl{1} & \ctrl{2} & \ctrl{4} & \targ & \ctrl{4} & \ctrl{2} & \ctrl{1} & \qw & \rstick{\ket{b_1 \oplus a}} \\
    \lstick{\ket{b_2}} & \targ \qwx[1] & \qw &&& \targ & \qw & \qw & \qw & \qw & \qw & \targ & \qw & \rstick{\ket{b_2 \oplus a}} \\
    \lstick{\ket{b_3}} & \targ \qwx[1] & \qw &&& \ctrl{1} & \targ & \qw & \qw & \qw & \targ & \ctrl{1} & \qw & \rstick{\ket{b_3 \oplus a}} \\
    \lstick{\ket{b_4}} & \targ \qwx[1] & \qw &\push{=}&& \targ & \qw & \qw & \qw & \qw & \qw & \targ & \qw & \rstick{\ket{b_4 \oplus a}} \\
    \lstick{\ket{b_5}} & \targ \qwx[1] & \qw &&& \ctrl{1} & \ctrl{2} & \targ & \qw & \targ & \ctrl{2} & \ctrl{1} & \qw & \rstick{\ket{b_5 \oplus a}} \\
    \lstick{\ket{b_6}} & \targ \qwx[1] & \qw &&& \targ & \qw & \qw & \qw & \qw & \qw & \targ & \qw & \rstick{\ket{b_6 \oplus a}} \\
    \lstick{\ket{b_7}} & \targ \qwx[1] & \qw &&& \ctrl{1} & \targ & \qw & \qw & \qw & \targ & \ctrl{1} & \qw & \rstick{\ket{b_7 \oplus a}} \\
    \lstick{\ket{b_8}} & \targ & \qw &&& \targ & \qw & \qw & \qw & \qw & \qw & \targ & \qw & \rstick{\ket{b_8 \oplus a}}
}
\end{flushleft}
\caption{The quantum circuit implementing the $F_8$ gate.}
\label{fig:F_8}
\end{figure}
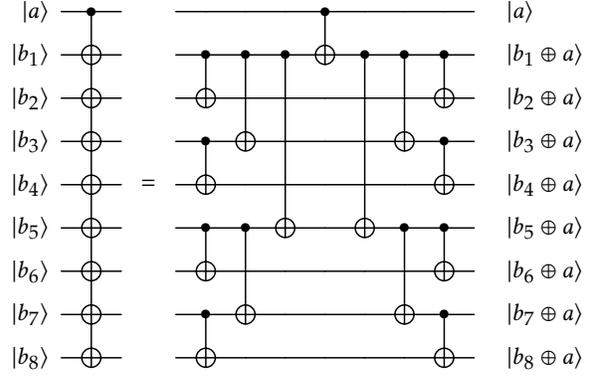

Note that the quantum fan-out gate $F_n$ is composed solely of CNOT gates, resulting in a $T$-depth of $0$.

\begin{theorem}
    \label{the:ReLU}
    There is a quantum circuit $C_n$ capable of implementing the $n$-bit ReLU function with a constant $T$-depth of $4$ without ancillary qubits. Furthermore, the circuit depth is $O(\log n)$, and the circuit size is $O(n)$. 
\end{theorem}

\begin{proof}
    The input $\boldsymbol{x} = x_1x_2\cdots x_n$, consisting of $n$ bits, is stored in $n$ qubits, with the sign bit represented by $x_1$. The $n-1$ initialized output qubits are denoted as $\ket{0}_1,\ket{0}_2,\dots,\ket{0}_{n-1}$. The construction of the quantum circuit $C_n$ proceeds in three steps:
    \begin{itemize}
        \item Firstly, apply an $X$ gate to the qubit $\ket{x_1}$.
        \item The second step involves applying the $n-1$ Toffoli gates. For the $i$-th Toffoli gate, the first control qubit is $\ket{x_1}$, the second control qubit is $\ket{x_{i+1}}$, and the target qubit is $\ket{0}_i$.
        \item Lastly, apply an $X$ gate to the qubit $\ket{x_1}$.
    \end{itemize}
    
    The quantum circuit implementing the ReLU function with input $\boldsymbol{x} \in \{0,1\}^5$ is shown in Fig.~\ref{fig:ReLU}.

    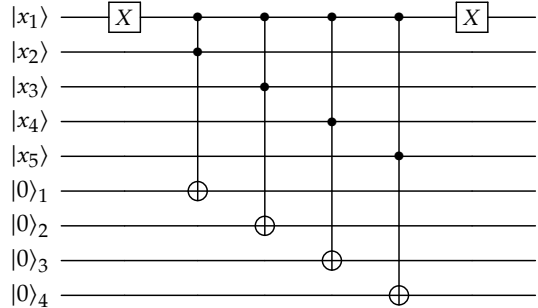
\begin{figure}[h]
    \centering
    \mbox{\Qcircuit @C=2em @R=0.2em @!R{
        \lstick{\ket{x_1}} & \gate{X} & \ctrl{1} & \ctrl{2} & \ctrl{3} & \ctrl{4} & \gate{X} & \qw \\
        \lstick{\ket{x_2}} & \qw & \ctrl{4} & \qw & \qw & \qw & \qw & \qw \\
        \lstick{\ket{x_3}} & \qw & \qw & \ctrl{4} & \qw & \qw & \qw & \qw \\
        \lstick{\ket{x_4}} & \qw & \qw & \qw & \ctrl{4} & \qw & \qw & \qw \\
        \lstick{\ket{x_5}} & \qw & \qw & \qw & \qw & \ctrl{4} & \qw & \qw \\
        \lstick{\ket{0}_1} & \qw & \targ & \qw & \qw & \qw & \qw & \qw \\
        \lstick{\ket{0}_2} & \qw & \qw & \targ & \qw & \qw & \qw & \qw \\
        \lstick{\ket{0}_3} & \qw & \qw & \qw & \targ & \qw & \qw & \qw \\
        \lstick{\ket{0}_4} & \qw & \qw & \qw & \qw & \targ & \qw & \qw
    }}
    \caption{The quantum circuit implementing the ReLU function with input $\boldsymbol{x} \in \{0,1\}^5$ includes input qubits $\ket{x_i}$ and output qubits $\ket{0}_j$.}
    \label{fig:ReLU}
    \end{figure}

        \begin{figure*}
    \centering
    \mbox{\Qcircuit @C=1em @R=1em @!R{
        & \ctrl{1} & \qw &&& \qw & \qw & \qw & \ctrl{2} & \qw & \qw & \qw & \ctrl{2} & \ctrl{1} & \gate{T} & \ctrl{1} & \qw \\
        & \ctrl{1} & \qw &\push{=}&& \qw & \ctrl{1} & \qw & \qw & \qw & \ctrl{1} & \gate{T} & \qw & \targ & \gate{T^{\dagger}} & \targ & \qw \\
        & \targ & \qw &&& \gate{H} & \targ & \gate{T^{\dagger}} & \targ & \gate{T} & \targ & \gate{T^{\dagger}} & \targ & \qw & \gate{T} & \gate{H} & \qw \gategroup{1}{9}{3}{9}{.7em}{--} \gategroup{1}{13}{3}{13}{.7em}{--} \gategroup{1}{14}{2}{14}{.7em}{--} \gategroup{1}{16}{2}{16}{.7em}{--}
    }}
    \caption{The quantum circuit for the Toffoli gate.}
    \label{fig:toffoli}
    \end{figure*}
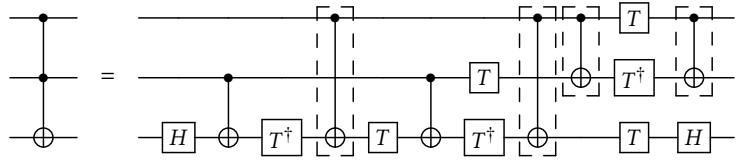

    \begin{figure*}
    \centering
    \mbox{\Qcircuit @C=1em @R=1em @!R{
        \lstick{\ket{x_1}} & \ctrl{1} & \ctrl{3} & \ctrl{5} & \ctrl{7} & \qw &&& \qw & \qw & \qw & \ctrl{2} & \qw & \qw & \qw & \ctrl{2} & \ctrl{1} & \gate{T^4} & \ctrl{1} & \qw & \qw \\
        \lstick{\ket{x_2}} & \ctrl{1} & \qw & \qw & \qw & \qw &&& \qw & \ctrl{1} & \qw & \qw & \qw & \ctrl{1} & \gate{T} & \qw & \targ \qwx[2] & \gate{T^{\dagger}} & \targ \qwx[2] & \qw & \qw \\
        \lstick{\ket{0}_1} & \targ & \qw & \qw & \qw & \qw &&& \gate{H} & \targ & \gate{T^{\dagger}} & \targ \qwx[2] & \gate{T} & \targ & \gate{T^{\dagger}} & \targ \qwx[2] & \qw & \gate{T} & \qw & \gate{H} & \qw \\
        \lstick{\ket{x_3}} & \qw & \ctrl{1} & \qw & \qw & \qw &&& \qw & \ctrl{1} & \qw & \qw & \qw & \ctrl{1} & \gate{T} & \qw & \targ \qwx[2] & \gate{T^{\dagger}} & \targ \qwx[2] & \qw & \qw \\
        \lstick{\ket{0}_2} & \qw & \targ & \qw & \qw & \qw &\push{=}&& \gate{H} & \targ & \gate{T^{\dagger}} & \targ \qwx[2] & \gate{T} & \targ & \gate{T^{\dagger}} & \targ \qwx[2] & \qw & \gate{T} & \qw & \gate{H} & \qw \\
        \lstick{\ket{x_4}} & \qw & \qw & \ctrl{1} & \qw & \qw &&& \qw & \ctrl{1} & \qw & \qw & \qw & \ctrl{1} & \gate{T} & \qw & \targ \qwx[2] & \gate{T^{\dagger}} & \targ \qwx[2] & \qw & \qw  \\
        \lstick{\ket{0}_3} & \qw & \qw & \targ & \qw & \qw &&& \gate{H} & \targ & \gate{T^{\dagger}} & \targ \qwx[2] & \gate{T} & \targ & \gate{T^{\dagger}} & \targ \qwx[2] & \qw & \gate{T} & \qw & \gate{H} & \qw \\
        \lstick{\ket{x_5}} & \qw & \qw & \qw & \ctrl{1} & \qw &&&  \qw & \ctrl{1} & \qw & \qw & \qw & \ctrl{1} & \gate{T} & \qw & \targ & \gate{T^{\dagger}} & \targ & \qw & \qw \\
        \lstick{\ket{0}_4} & \qw & \qw & \qw & \targ & \qw &&& \gate{H} & \targ & \gate{T^{\dagger}} & \targ & \gate{T} & \targ & \gate{T^{\dagger}} & \targ & \qw & \gate{T} & \qw & \gate{H} & \qw \gategroup{1}{12}{9}{12}{.7em}{--} \gategroup{1}{16}{9}{16}{.7em}{--} \gategroup{1}{17}{8}{17}{.7em}{--}
        \gategroup{1}{19}{8}{19}{.7em}{--}
        \gategroup{1}{18}{1}{18}{.7em}{.}
    }}
    \caption{The quantum circuit for the four shared-control Toffoli gates.}
    \label{fig:4Toffoli}
    \end{figure*}
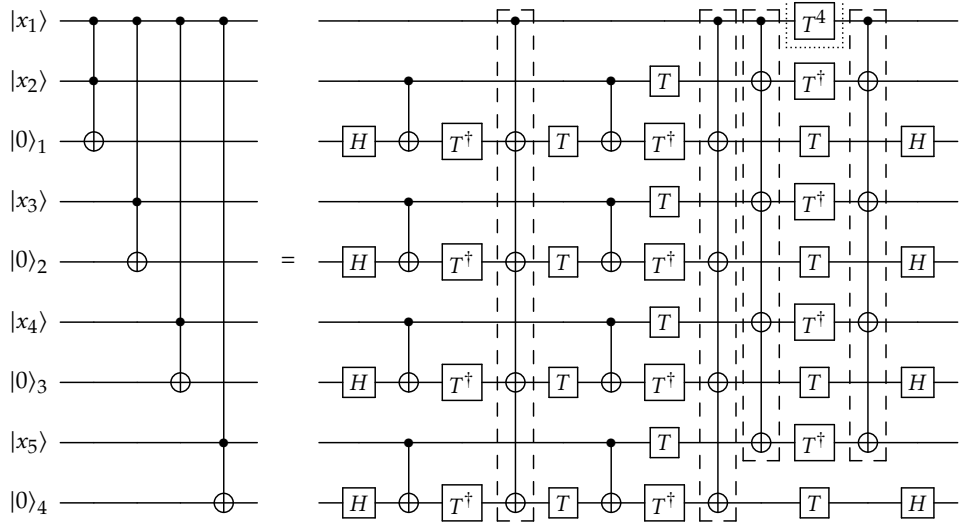

    To confirm the correctness of our circuit, let us examine the following scenarios. When $x_1 = 1$, all the output qubits remain unaffected, indicating a negative input. Conversely, when $x_1 = 0$, the $n-1$ Toffoli gates copy the values of $\ket{x_{i+1}}$ to $\ket{0}_i$ for $i = 1,2,\dots, n-1$. Consequently, we effectively obtain $\max(\boldsymbol{x},0)$.

    Next, we demonstrate that the $n-1$ Toffoli gates can be implemented using a quantum circuit with a $T$-depth of 4. The widely recognized quantum circuit for the Toffoli gate is depicted in Fig.~\ref{fig:toffoli} \cite{nielsen2010quantum}. It is important to note that the $n-1$ Toffoli gates in $C_n$ share a common control qubit, $\ket{x_1}$, while the other control qubits and target qubits are all distinct. We employ the shared-control Toffoli technique proposed in \cite{gokhale2021quantum}. For visualization, the implementation of 4 Toffoli gates is depicted in Fig.~\ref{fig:4Toffoli}. 

    Note that in the circuit of the Toffoli gate, the CNOT gates with $\ket{x_1}$ as the control qubit are replaced with quantum fan-out gates in the circuit of shared-control Toffoli gates, as illustrated within the dashed box in Fig.~\ref{fig:toffoli} and Fig.~\ref{fig:4Toffoli}. Additionally, for each Toffoli gate, a $T$ gate is applied to $\ket{x_1}$, such as the $T^4$ gate shown in Fig.~\ref{fig:4Toffoli}.

    To analyze the circuit complexity, two key observations are noted. Firstly, as stated in Lemma~\ref{lem:fanout}, the quantum fan-out gate does not contribute to the $T$-depth. Secondly, considering that the $T$-depth of the $T^{n-1}$ gate is not greater than 1 due to $T^2 = S$, where the $S$ gate is a Clifford gate. Hence, the $T$-depth of our circuit implementing the ReLU function is 4. According to Lemma~\ref{lem:fanout}, its circuit depth and size are $O(\log n)$ and $O(n)$ respectively.
\end{proof}

Now, we establish the lower bound on the depth and size of the quantum circuit implementing the ReLU function. Our approach is inspired by Theorem 3.2 in \cite{fang2006quantum}.

\begin{theorem}
    \label{the:ReLUl}
    If a quantum circuit $C$ implements the $n$-bit ReLU function, then the depth and size of $C$ are lower bounded by $\Omega(\log n)$ and $\Omega(n)$, respectively.
\end{theorem}

\begin{proof}
    The size lower bound is directly derived from the fact that the output is influenced by all input bits. Regarding the depth lower bound, the key observation in our proof is that the sign bit $x_1$ could affect each output bit. We demonstrate that if one qubit can influence the measurement result of $n-1$ qubits, then the circuit depth has a lower bound of $\Omega(\log n)$.

    We assume the existence of a quantum circuit $C$ that implements the ReLU function. This implies that for a ReLU function $R$, the following equality holds:
    \begin{align*}
        &\bra{x_1,\dots,x_n,0,\dots,0} C^{\dagger} M_i C \ket{x_1,\dots,x_n,0,\dots,0} \\
        = &\bra{x_1,\dots,x_n,0,\dots,0} R^{\dagger} M_i R \ket{x_1,\dots,x_n,0,\dots,0}.
    \end{align*}
    Here, $i \in \{1,2,\dots, n-1\}$ and $M_i$ represent an observable on the $i$-th output qubit.

    Then, we consider the number of qubits that can influence the measurement result of $M_i$ in a $d$-depth quantum circuit. Let us assume that the quantum circuit $C$ consists of $d$ layers, and it can be represented as $C = L_1 L_2 \dots L_d$, where each $L_j$ represents a layer of quantum gates consisting of Clifford + $T$ gates. For each $d$ and $i$, we can identify minimal subsets $L_{1,i}',L_{2,i}',\dots,L_{d,i}'$ such that the following equality holds: 
    \begin{align*}
        & L_{d}^{\dagger} L_{d-1}^{\dagger} \cdots L_{1}^{\dagger} M_i L_{1} \cdots L_{d-1} L_{d} \\
        = & L_{d,i}'^{\dagger} L_{d-1,i}'^{\dagger} \cdots L_{1,i}'^{\dagger} M_i L_{1,i}' \cdots L_{d-1,i}' L_{d,i}'.
    \end{align*}
    We denote $S_{d,i}$ as the set of qubits involved in the quantum gates belonging to $L_{d,i}'$. It is important to note that the sign bit $x_1$ should belong to all sets $S_{d,i}$ for $i\in \{1,2,\dots,n-1\}$. If it does not, then the sign bit cannot influence the measurement result of $M_i$, which leads to a contradiction.

    Let $a_d$ denote the qubit that appears most frequently in $S_{d,1},S_{d,2},\dots,S_{d,n-1}$. Our aim is to prove that the number of sets $S_{d,i}$ containing $a_d$ is no more than $2^{d-1}$. We prove this statement by induction on $d$.

    For $d=1$, each qubit appears in at most one set. This is because there is only one layer of quantum gates and $L_{1,i}$ contains only the quantum gate applied to the $i$-th output qubit. The other gates are canceled, resulting in $L_{1}^{\dagger} M_i L_{1} = L_{1,i}'^{\dagger} M_i L_{1,i}'$.

    Now, suppose our statement holds for $d$. We observe that there are only two ways for a qubit $a$ to be contained in $S_{d+1,i}$.

    The first way is if $a$ is already present in $S_{d,i}$, then $a$ will also be present in $S_{d+1,i}$ since $S_{d,i} \subset S_{d+1,i}$.

    The second way is if there is a CNOT gate applied between qubits $a$ and $b$ belonging to $L_{d+1}$, and $b$ is in $S_{d,i}$. Both $a$ and $b$ appear in no more than $2^{d-1}$ different sets $S_{d,i}$ according to our supposition. Therefore, $a$ appears in no more than $2^d$ different sets $S_{d+1,i}$.

    Finally, since the sign bit needs to influence all the qubits for the output, we require $2^{d-1} \geq n-1$. Therefore, the circuit depth $d$ has a lower bound of $\Omega(\log n)$.
\end{proof}

\subsection{Implementing the ReLU Function on a 2D Grid}

    \begin{figure}[h]
    \centering
    \begin{tikzpicture}[
    roundnode/.style={circle, draw=black!80, fill=black!80, very thick, minimum size=5mm, text width=1em, text height=1ex, text=white}, node distance=0.2cm and 0.2cm]
    \node[roundnode]  (x1)                {$x_1$};
    \node[roundnode]  (x2)  [below=of x1] {$x_2$};
    \node[roundnode]  (x3)  [below=of x2] {$x_3$};
    \node[roundnode]  (x4)  [below=of x3] {$x_4$};
    \node[roundnode]  (x5)  [below=of x4] {$x_5$};
    \node[roundnode]  (x6)  [below=of x5] {$x_6$};
    \node[roundnode]  (x7)  [below=of x6] {$x_7$};
    \node[roundnode]  (x8)  [below=of x7] {$x_8$};

    \node[roundnode]  (a0)  [right=of x1] {$a_0$};
    \node[roundnode]  (a1)  [right=of x2] {$a_1$};
    \node[roundnode]  (a2)  [right=of x3] {$a_2$};
    \node[roundnode]  (a3)  [right=of x4] {$a_3$};
    \node[roundnode]  (a4)  [right=of x5] {$a_4$};
    \node[roundnode]  (a5)  [right=of x6] {$a_5$};
    \node[roundnode]  (a6)  [right=of x7] {$a_6$};
    \node[roundnode]  (a7)  [right=of x8] {$a_7$};

    \node[roundnode]  (x9)  [right=of a0] {$x_9$};
    \node[roundnode]  (x10)  [right=of a1] {$x_{10}$};
    \node[roundnode]  (x11)  [right=of a2] {$x_{11}$};
    \node[roundnode]  (x12)  [right=of a3] {$x_{12}$};
    \node[roundnode]  (x13)  [right=of a4] {$x_{13}$};
    \node[roundnode]  (x14)  [right=of a5] {$x_{14}$};
    \node[roundnode]  (x15)  [right=of a6] {$x_{15}$};
    \node[roundnode]  (x16)  [right=of a7] {$x_{16}$};

    \node[roundnode]  (a8)  [right=of x9] {$a_8$};
    \node[roundnode]  (a9)  [right=of x10] {$a_9$};
    \node[roundnode]  (a10)  [right=of x11] {$a_{10}$};
    \node[roundnode]  (a11)  [right=of x12] {$a_{11}$};
    \node[roundnode]  (a12)  [right=of x13] {$a_{12}$};
    \node[roundnode]  (a13)  [right=of x14] {$a_{13}$};
    \node[roundnode]  (a14)  [right=of x15] {$a_{14}$};
    \node[roundnode]  (a15)  [right=of x16] {$a_{15}$};

    \node[roundnode]  (x17)  [right=of a8] {$x_{17}$};
    \node[roundnode]  (x18)  [right=of a9] {$x_{18}$};
    \node[roundnode]  (x19)  [right=of a10] {$x_{19}$};
    \node[roundnode]  (x20)  [right=of a11] {$x_{20}$};
    \node[roundnode]  (x21)  [right=of a12] {$x_{21}$};
    \node[roundnode]  (x22)  [right=of a13] {$x_{22}$};
    \node[roundnode]  (x23)  [right=of a14] {$x_{23}$};
    \node[roundnode]  (x24)  [right=of a15] {$x_{24}$};

    \node[roundnode]  (a16)  [right=of x17] {$a_{16}$};
    \node[roundnode]  (a17)  [right=of x18] {$a_{17}$};
    \node[roundnode]  (a18)  [right=of x19] {$a_{18}$};
    \node[roundnode]  (a19)  [right=of x20] {$a_{19}$};
    \node[roundnode]  (a20)  [right=of x21] {$a_{20}$};
    \node[roundnode]  (a21)  [right=of x22] {$a_{21}$};
    \node[roundnode]  (a22)  [right=of x23] {$a_{22}$};
    \node[roundnode]  (a23)  [right=of x24] {$a_{23}$};

    \node[roundnode]  (x25)  [right=of a16] {$x_{25}$};
    \node[roundnode]  (x26)  [right=of a17] {$x_{26}$};
    \node[roundnode]  (x27)  [right=of a18] {$x_{27}$};
    \node[roundnode]  (x28)  [right=of a19] {$x_{28}$};
    \node[roundnode]  (x29)  [right=of a20] {$x_{29}$};
    \node[roundnode]  (x30)  [right=of a21] {$x_{30}$};
    \node[roundnode]  (x31)  [right=of a22] {$x_{31}$};
    \node[roundnode]  (x32)  [right=of a23] {$x_{32}$};

    \node[roundnode]  (a24)  [right=of x25] {$a_{24}$};
    \node[roundnode]  (a25)  [right=of x26] {$a_{25}$};
    \node[roundnode]  (a26)  [right=of x27] {$a_{26}$};
    \node[roundnode]  (a27)  [right=of x28] {$a_{27}$};
    \node[roundnode]  (a28)  [right=of x29] {$a_{28}$};
    \node[roundnode]  (a29)  [right=of x30] {$a_{29}$};
    \node[roundnode]  (a30)  [right=of x31] {$a_{30}$};
    \node[roundnode]  (a31)  [right=of x32] {$a_{31}$};
    
    \draw[-, line width=4pt, black] (x1.east) -- (a0.west);
    \draw[-, line width=4pt, black] (x2.east) -- (a1.west);
    \draw[-, line width=4pt, black] (x3.east) -- (a2.west);
    \draw[-, line width=4pt, black] (x4.east) -- (a3.west);
    \draw[-, line width=4pt, black] (x5.east) -- (a4.west);
    \draw[-, line width=4pt, black] (x6.east) -- (a5.west);
    \draw[-, line width=4pt, black] (x7.east) -- (a6.west);
    \draw[-, line width=4pt, black] (x8.east) -- (a7.west);

    \draw[-, line width=4pt, black] (x9.east) -- (a8.west);
    \draw[-, line width=4pt, black] (x10.east) -- (a9.west);
    \draw[-, line width=4pt, black] (x11.east) -- (a10.west);
    \draw[-, line width=4pt, black] (x12.east) -- (a11.west);
    \draw[-, line width=4pt, black] (x13.east) -- (a12.west);
    \draw[-, line width=4pt, black] (x14.east) -- (a13.west);
    \draw[-, line width=4pt, black] (x15.east) -- (a14.west);
    \draw[-, line width=4pt, black] (x16.east) -- (a15.west);

    \draw[-, line width=4pt, black] (x17.east) -- (a16.west);
    \draw[-, line width=4pt, black] (x18.east) -- (a17.west);
    \draw[-, line width=4pt, black] (x19.east) -- (a18.west);
    \draw[-, line width=4pt, black] (x20.east) -- (a19.west);
    \draw[-, line width=4pt, black] (x21.east) -- (a20.west);
    \draw[-, line width=4pt, black] (x22.east) -- (a21.west);
    \draw[-, line width=4pt, black] (x23.east) -- (a22.west);
    \draw[-, line width=4pt, black] (x24.east) -- (a23.west);

    \draw[-, line width=4pt, black] (x25.east) -- (a24.west);
    \draw[-, line width=4pt, black] (x26.east) -- (a25.west);
    \draw[-, line width=4pt, black] (x27.east) -- (a26.west);
    \draw[-, line width=4pt, black] (x28.east) -- (a27.west);
    \draw[-, line width=4pt, black] (x29.east) -- (a28.west);
    \draw[-, line width=4pt, black] (x30.east) -- (a29.west);
    \draw[-, line width=4pt, black] (x31.east) -- (a30.west);
    \draw[-, line width=4pt, black] (x32.east) -- (a31.west);

    \draw[-, line width=4pt, black] (a0.east) -- (x9.west);
    \draw[-, line width=4pt, black] (a1.east) -- (x10.west);
    \draw[-, line width=4pt, black] (a2.east) -- (x11.west);
    \draw[-, line width=4pt, black] (a3.east) -- (x12.west);
    \draw[-, line width=4pt, black] (a4.east) -- (x13.west);
    \draw[-, line width=4pt, black] (a5.east) -- (x14.west);
    \draw[-, line width=4pt, black] (a6.east) -- (x15.west);
    \draw[-, line width=4pt, black] (a7.east) -- (x16.west);

    \draw[-, line width=4pt, black] (a8.east) -- (x17.west);
    \draw[-, line width=4pt, black] (a9.east) -- (x18.west);
    \draw[-, line width=4pt, black] (a10.east) -- (x19.west);
    \draw[-, line width=4pt, black] (a11.east) -- (x20.west);
    \draw[-, line width=4pt, black] (a12.east) -- (x21.west);
    \draw[-, line width=4pt, black] (a13.east) -- (x22.west);
    \draw[-, line width=4pt, black] (a14.east) -- (x23.west);
    \draw[-, line width=4pt, black] (a15.east) -- (x24.west);

    \draw[-, line width=4pt, black] (a16.east) -- (x25.west);
    \draw[-, line width=4pt, black] (a17.east) -- (x26.west);
    \draw[-, line width=4pt, black] (a18.east) -- (x27.west);
    \draw[-, line width=4pt, black] (a19.east) -- (x28.west);
    \draw[-, line width=4pt, black] (a20.east) -- (x29.west);
    \draw[-, line width=4pt, black] (a21.east) -- (x30.west);
    \draw[-, line width=4pt, black] (a22.east) -- (x31.west);
    \draw[-, line width=4pt, black] (a23.east) -- (x32.west);

    \draw[-, line width=4pt, black] (x1.south) -- (x2.north);
    \draw[-, line width=4pt, black] (x2.south) -- (x3.north);
    \draw[-, line width=4pt, black] (x3.south) -- (x4.north);
    \draw[-, line width=4pt, black] (x4.south) -- (x5.north);
    \draw[-, line width=4pt, black] (x5.south) -- (x6.north);
    \draw[-, line width=4pt, black] (x6.south) -- (x7.north);
    \draw[-, line width=4pt, black] (x7.south) -- (x8.north);

    \draw[-, line width=4pt, black] (x9.south) -- (x10.north);
    \draw[-, line width=4pt, black] (x10.south) -- (x11.north);
    \draw[-, line width=4pt, black] (x11.south) -- (x12.north);
    \draw[-, line width=4pt, black] (x12.south) -- (x13.north);
    \draw[-, line width=4pt, black] (x13.south) -- (x14.north);
    \draw[-, line width=4pt, black] (x14.south) -- (x15.north);
    \draw[-, line width=4pt, black] (x15.south) -- (x16.north);

    \draw[-, line width=4pt, black] (x17.south) -- (x18.north);
    \draw[-, line width=4pt, black] (x18.south) -- (x19.north);
    \draw[-, line width=4pt, black] (x19.south) -- (x20.north);
    \draw[-, line width=4pt, black] (x20.south) -- (x21.north);
    \draw[-, line width=4pt, black] (x21.south) -- (x22.north);
    \draw[-, line width=4pt, black] (x22.south) -- (x23.north);
    \draw[-, line width=4pt, black] (x23.south) -- (x24.north);

    \draw[-, line width=4pt, black] (x25.south) -- (x26.north);
    \draw[-, line width=4pt, black] (x26.south) -- (x27.north);
    \draw[-, line width=4pt, black] (x27.south) -- (x28.north);
    \draw[-, line width=4pt, black] (x28.south) -- (x29.north);
    \draw[-, line width=4pt, black] (x29.south) -- (x30.north);
    \draw[-, line width=4pt, black] (x30.south) -- (x31.north);
    \draw[-, line width=4pt, black] (x31.south) -- (x32.north);

    \draw[-, line width=4pt, black] (a0.south) -- (a1.north);
    \draw[-, line width=4pt, black] (a1.south) -- (a2.north);
    \draw[-, line width=4pt, black] (a2.south) -- (a3.north);
    \draw[-, line width=4pt, black] (a3.south) -- (a4.north);
    \draw[-, line width=4pt, black] (a4.south) -- (a5.north);
    \draw[-, line width=4pt, black] (a5.south) -- (a6.north);
    \draw[-, line width=4pt, black] (a6.south) -- (a7.north);

    \draw[-, line width=4pt, black] (a8.south) -- (a9.north);
    \draw[-, line width=4pt, black] (a9.south) -- (a10.north);
    \draw[-, line width=4pt, black] (a10.south) -- (a11.north);
    \draw[-, line width=4pt, black] (a11.south) -- (a12.north);
    \draw[-, line width=4pt, black] (a12.south) -- (a13.north);
    \draw[-, line width=4pt, black] (a13.south) -- (a14.north);
    \draw[-, line width=4pt, black] (a14.south) -- (a15.north);

    \draw[-, line width=4pt, black] (a16.south) -- (a17.north);
    \draw[-, line width=4pt, black] (a17.south) -- (a18.north);
    \draw[-, line width=4pt, black] (a18.south) -- (a19.north);
    \draw[-, line width=4pt, black] (a19.south) -- (a20.north);
    \draw[-, line width=4pt, black] (a20.south) -- (a21.north);
    \draw[-, line width=4pt, black] (a21.south) -- (a22.north);
    \draw[-, line width=4pt, black] (a22.south) -- (a23.north);

    \draw[-, line width=4pt, black] (a24.south) -- (a25.north);
    \draw[-, line width=4pt, black] (a25.south) -- (a26.north);
    \draw[-, line width=4pt, black] (a26.south) -- (a27.north);
    \draw[-, line width=4pt, black] (a27.south) -- (a28.north);
    \draw[-, line width=4pt, black] (a28.south) -- (a29.north);
    \draw[-, line width=4pt, black] (a29.south) -- (a30.north);
    \draw[-, line width=4pt, black] (a30.south) -- (a31.north);

    \draw[line width=2pt, red] (-0.6,-3.2) -- (7,-3.2);
    \draw[line width=2pt, red] (3.2,0.6) -- (3.2,-7);
    \end{tikzpicture}
    \caption{An illustration of qubit arrangement on a 2D grid, with $x_i$ representing input qubits and $a_j$ representing output qubits. The main idea of constructing the circuit is to partition the problem into four smaller ones, as illustrated in the figure.}
    \label{fig:2Dgrid1}
    \end{figure}
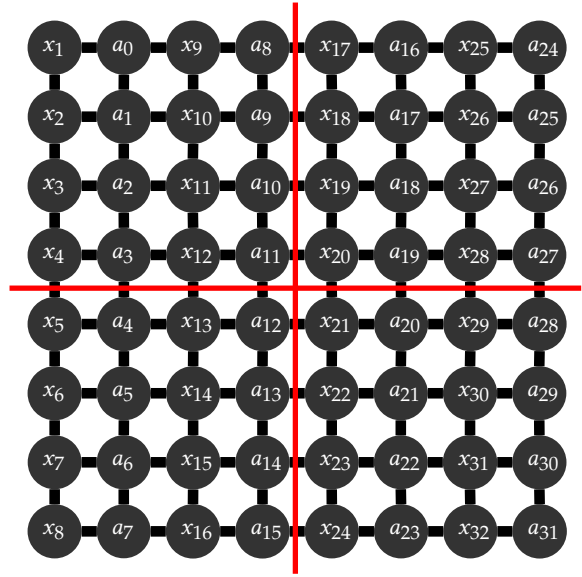

    \begin{figure*}
    \centering
    \includegraphics[width=0.85\textwidth]{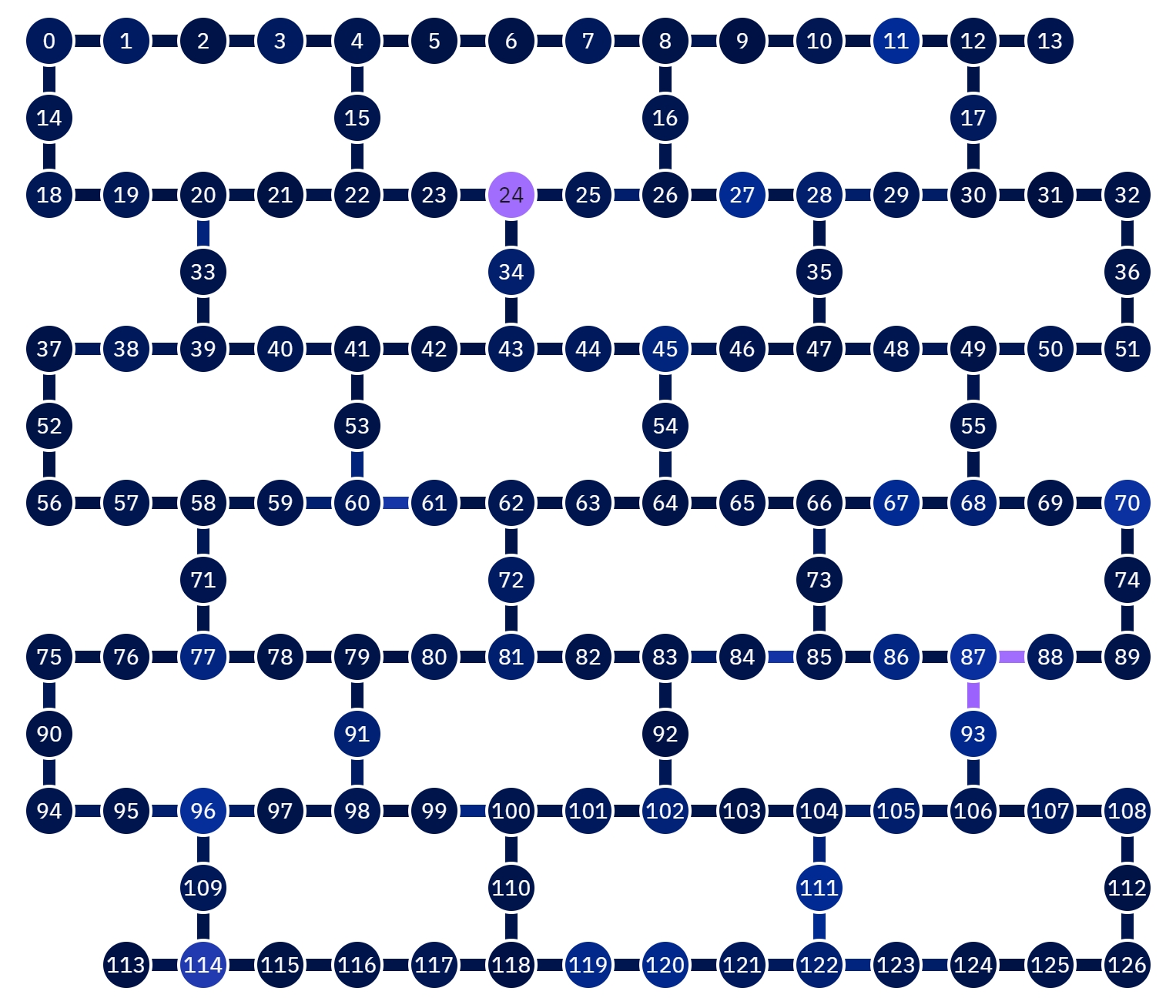}
    \caption{The qubits arrangement of 127 qubits quantum processor \texttt{ibm\_brisbane} from IBM.}
    \label{fig:ibmq}
    \end{figure*}

In this subsection, our focus is on implementing the ReLU function using a quantum circuit on a 2D grid. We begin by detailing the construction of the circuit, followed by establishing a lower bound for implementing ReLU on a 2D grid. Finally, we demonstrate the extension of this approach to other planar structures.

\begin{theorem}
    There exists a quantum circuit $C_n$ capable of implementing the $n$-bit ReLU function with a constant $T$-depth of $4$ without ancillary qubits. Furthermore, the circuit satisfies the connectivity constraint of a 2D grid, with the depth and size being $O(\sqrt{n})$ and $O(n)$ respectively.
\end{theorem}

\begin{proof}
    The quantum circuit presented here is derived from the circuit outlined in Theorem~\ref{the:ReLU}. Our approach involves strategically positioning each qubit to facilitate the direct application of all CNOT gates, with the exception of four quantum fan-out gates. Subsequently, we devise a distinct quantum circuit consisting solely of CNOT gates to execute the quantum fan-out operation within a 2D grid. Consequently, the $T$-depth of our circuit remains at 4.

    We assume that the number of input qubits is given by $n=k^2/2$ for some even integer $k$. Cases where $n$ assumes other values can be derived from this scenario. In a 2D grid layout, we position the input qubits $x_1,\dots,x_n$ along the odd columns, while the output qubits $a_1,\dots,a_{n-1}$ are situated across the even columns. Specifically, for each output qubit $a_i$, we position $x_{i+1}$ to its immediate left. An illustration of this configuration for $n=32$ is depicted in Fig.~\ref{fig:2Dgrid1}, where the unused qubit $a_0$ is added for aesthetic purposes.


    Note that on the 2D grid, the CNOT gate can be applied directly between two adjacent qubits. Therefore, according to Theorem~\ref{the:ReLU}, all the CNOT gates in the ReLU circuit can be applied directly, and our focus shifts solely to implementing the quantum fan-out gates on a 2D grid. Consider the quantum fan-out gate circuit depicted in Fig.~\ref{fig:F_8}. Initially, we apply a CNOT gate to add the value of $\ket{a}$ to $\ket{b_1}$. Subsequently, by utilizing CNOT gates on both sides of $\ket{b_1}$, we extend the value of $\ket{a}$ to $\ket{b_5}$. Leveraging $\ket{b_1}$ and $\ket{b_5}$, we concurrently add the value of $\ket{a}$ to $\ket{b_3}$ and $\ket{b_7}$, achieving parallelism. We adopt a similar approach on the 2D grid, dividing it evenly into four parts and transferring the value of $\ket{x_1}$ to the four smaller 2D grids. This partitions the problem into four smaller sub-problems that can be processed concurrently. We then proceed with recursion. In Fig.~\ref{fig:2Dgrid1}, we illustrate an example where we first transfer the value of $\ket{x_1}$ to $\ket{x_{17}}$, and then concurrently transfer the value of $\ket{x_1}$ to $\ket{x_{5}}$ and $\ket{x_{21}}$, thus partitioning the problem into four smaller ones. To implement these long-distance CNOT gates, we utilize swap gates, which can be achieved using three CNOT gates. Denoting $D(n)$ and $S(n)$ as the depth and size of this circuit when the number of input bits is $n=k^2/2$, respectively, we establish the following recursive relation:
    \begin{align*}
        D(\frac{k^2}2) & \leq D(\frac{(k/2)^2}{2}) + O(k),\\
        S(\frac{k^2}2) & \leq 4S(\frac{(k/2)^2}{2}) + O(k). 
    \end{align*}
    Then, by the master theorem \cite{cormen2022introduction}, we have $D(n) = O(\sqrt{n})$ and $S(n) = O(n)$. We only present the quantum fan-out gate for $x_i$. The implementation for $a_i$ follows a similar pattern and is omitted. If $n$ is not equal to $k^2/2$ for some integer $k$, the circuit remains similar. The crucial step is to divide the 2D grid into four smaller parts and apply recursion. This concludes the proof.

\end{proof}

\begin{theorem}
    Suppose a quantum circuit $C$ implements the $n$-bit ReLU function, restricting CNOT gates in $C$ to be applied only between nearby qubits on a 2D grid. In that case, the depth and size of $C$ are lower bounded by $\Omega(\sqrt{n})$ and $\Omega(n)$, respectively.
\end{theorem}

\begin{proof}
    The size lower bound directly follows from Theorem~\ref{the:ReLUl}. As for the depth lower bound, similar to the proof of Theorem~\ref{the:ReLUl}, we assume that the quantum circuit $C$ has $d$ layers, represented as $C = L_1 L_2 \dots L_d$, where each $L_j$ represents a layer of quantum gates consisting of Clifford + $T$ gates. For each $d$ and $i$, we can find the minimal subset $L_{1,i}',L_{2,i}',\dots,L_{d,i}'$ such that the following equality holds:
    \begin{align*}
        & L_{d}^{\dagger} L_{d-1}^{\dagger} \cdots L_{1}^{\dagger} M_i L_{1} \cdots L_{d-1} L_{d} \\
        = & L_{d,i}'^{\dagger} L_{d-1,i}'^{\dagger} \cdots L_{1,i}'^{\dagger} M_i L_{1,i}' \cdots L_{d-1,i}' L_{d,i}'.
    \end{align*}
    Here, $i \in \{1,2,\dots, n-1\}$, and $M_i$ represents an observable on the $i$-th output qubit. We denote $S_{d,i}$ as the set of qubits involved in the quantum gates belonging to $L_{d,i}'$.
    
    The crucial observation is that if a qubit $a$ is part of $S_{d+1,i}$ but not $S_{d,i}$, it indicates the presence of a CNOT gate in $L_{d+1,i}$ that acts on qubits $a$ and $b$, where qubit $b$ is in $S_{d,i}$. This implies that the minimum distance between the sign bit $x_1$ and all the output qubits $a_1,a_2,\dots,a_{n-1}$ should be at most $d$ in the 2D grid. Otherwise, $S_{d,i}$ cannot include $x_1$.

    Consider that within a 2D grid, the number of points within a distance of $d$ or less is $2d(d+1)$. To ensure that this quantity is greater than or equal to $n-1$, we need $2d(d+1) \geq n-1$. Consequently, the circuit depth is lower bounded by $\Omega(\sqrt{n})$.
\end{proof}

Our approach to designing a circuit for implementing the ReLU function on a 2D grid can be readily extended to other planar structures, such as the Eagle r3 type quantum processor from IBM, as depicted in Fig.~\ref{fig:ibmq}. In this type of quantum processor, the ReLU function can be implemented with a constant $T$-depth, a circuit depth of $\Theta(\sqrt{n})$, and a size of $\Theta(n)$. We can position each output qubit $a_i$ near its corresponding input qubit $x_{i+1}$. Employing a similar recursion method, we achieve an $O(\sqrt{n})$ circuit depth. The lower bound of $\Omega(\sqrt{n})$ for the circuit depth directly follows from the fact that the number of qubits within a distance no greater than $d$ from a fixed qubit is upper-bounded by $O(d^2)$.

\section{Implementation of Leaky ReLU function}
\label{sec:leakyrelu}

The definition of the Leaky ReLU function is as follows:
\begin{equation*}
    f(x) = \left\{
    \begin{array}{cc}
        x & x \geq 0 \\
        \alpha x & x < 0, \ 0 < \alpha < 1
    \end{array}
    \right..
\end{equation*}

\begin{definition}
    A circuit $C_n$ can implement the $n$-bit Leaky ReLU function, which means that:
    \begin{equation*}
        C_n\ket{\boldsymbol{x}}\ket{0}_{m} = \ket{\boldsymbol{x}}\ket{\max(\boldsymbol{x},\alpha \boldsymbol{x})}.
    \end{equation*}
    Where $\boldsymbol{x} \in \{0,1\}^n$ represents a fixed-point number, and the circuit $C_n$ consists of Clifford+$T$ gates. Here, $\alpha$ is a predetermined parameter chosen from the set $\{0.125, 0.0625, 0.03125, 0.015625\}$. The number of output qubits $m$ is determined by the number of bits required to store $\max(\boldsymbol{x},\alpha \boldsymbol{x})$.
\end{definition}

Note that we choose $\alpha$ from the set $\{2^{-3},2^{-4},2^{-5},2^{-6}\}$. Therefore, $m = n+3, n+4, n+5,$ or $n+6$ is sufficient to store the result $\max(x, \alpha x)$ exactly, depending on the specific choice of $\alpha$. The input and output can be coded using either 2's complement or True form. The quantum circuit that implements the Leaky ReLU function will vary depending on the chosen coding method. For example, if $x$ is coded as $1011$, then $0.125x$ should be coded as $1000.011$ and $1111.011$ in True form and 2's complement, respectively. Note that the $n$ and $m$ qubits are needed to store the input and output, respectively. When the input has $n$ bits, we will use $\ket{0}_1\ket{0}_2 \cdots \ket{0}_n$ to store the integer part and $\ket{0}_{n+1} \cdots \ket{0}_m$ to store the decimal part.

\begin{figure*}
    \centering
    \mbox{\Qcircuit @C=1.7em @R=0.2em @!R{
        \lstick{\ket{x_1}} & \gate{X} & \ctrl{1} & \ctrl{2} & \ctrl{3} & \ctrl{4} & \gate{X} \barrier[-1.1em]{12} & \ctrl{1} & \ctrl{2} & \ctrl{3} & \ctrl{4} & \ctrl{5} & \ctrl{6} & \ctrl{7} & \ctrl{8} & \qw \\
        \lstick{\ket{x_2}} & \qw & \ctrl{5} & \qw & \qw & \qw & \qw & \ctrl{8} & \qw & \qw & \qw & \qw & \qw & \qw & \qw & \qw \\
        \lstick{\ket{x_3}} & \qw & \qw & \ctrl{5} & \qw & \qw & \qw & \qw & \ctrl{8} & \qw & \qw & \qw & \qw & \qw & \qw & \qw \\
        \lstick{\ket{x_4}} & \qw & \qw & \qw & \ctrl{5} & \qw & \qw & \qw & \qw & \ctrl{8} & \qw & \qw & \qw & \qw & \qw & \qw \\
        \lstick{\ket{x_5}} & \qw & \qw & \qw & \qw & \ctrl{5} & \qw & \qw & \qw & \qw & \ctrl{8} & \qw & \qw & \qw & \qw & \qw \\
        \lstick{\ket{0}_1} & \qw & \qw & \qw & \qw & \qw & \qw & \qw & \qw & \qw & \qw & \targ & \qw & \qw & \qw & \qw \\
        \lstick{\ket{0}_2} & \qw & \targ & \qw & \qw & \qw & \qw & \qw & \qw & \qw & \qw & \qw & \targ & \qw & \qw & \qw \\
        \lstick{\ket{0}_3} & \qw & \qw & \targ & \qw & \qw & \qw & \qw & \qw & \qw & \qw & \qw & \qw & \targ & \qw & \qw \\
        \lstick{\ket{0}_4} & \qw & \qw & \qw & \targ & \qw & \qw & \qw & \qw & \qw & \qw & \qw & \qw & \qw & \targ & \qw \\
        \lstick{\ket{0}_5} & \qw & \qw & \qw & \qw & \targ & \qw & \targ & \qw & \qw & \qw & \qw & \qw & \qw & \qw & \qw \\
        \lstick{\ket{0}_6} & \qw & \qw & \qw & \qw & \qw & \qw & \qw & \targ & \qw & \qw & \qw & \qw & \qw & \qw & \qw \\
        \lstick{\ket{0}_7} & \qw & \qw & \qw & \qw & \qw & \qw & \qw & \qw & \targ & \qw & \qw & \qw & \qw & \qw & \qw \\
        \lstick{\ket{0}_8} & \qw & \qw & \qw & \qw & \qw & \qw & \qw & \qw & \qw & \targ & \qw & \qw & \qw & \qw & \qw \gategroup{1}{13}{9}{15}{1.7em}{--}
    }}
    \caption{The quantum circuit for the $5$-bit Leaky ReLU function. Here, $\ket{x_i}$ represents input qubits, and $\ket{0}_j$ represents output qubits.}
    \label{fig:lReLU}
    \end{figure*}
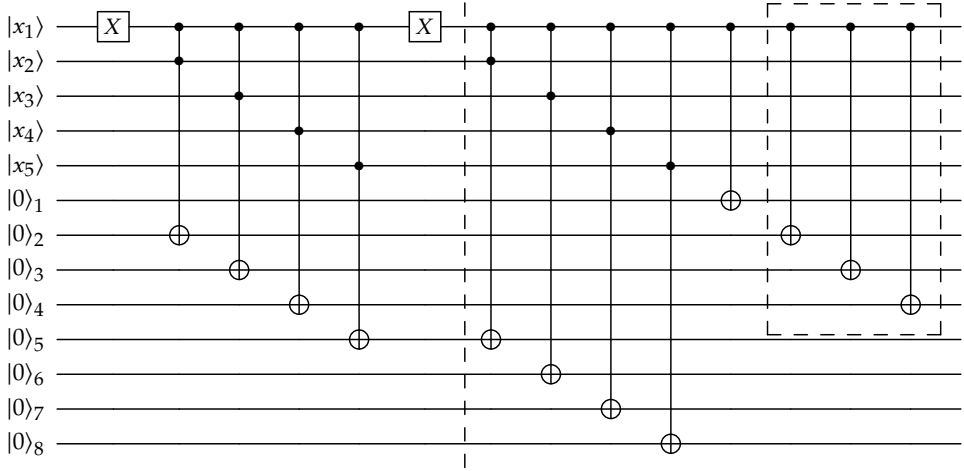

\begin{theorem}
    There exists a quantum circuit $C_n$ capable of implementing the $n$-bit Leaky ReLU function with a constant $T$-depth of $8$ without ancillary qubits. Additionally, the circuit depth and size are $O(\log n)$ and $O(n)$, respectively.
\end{theorem}

\begin{proof}
    To implement the Leaky ReLU function, we output $x$ if the sign bit is $0$, and we output $\alpha x$ if the sign bit is $1$. It is worth noting that we choose $\alpha$ from the set $\{0.125, 0.0625, 0.03125, 0.015625\} = \{2^{-3},2^{-4},2^{-5},2^{-6}\}$. Therefore, $\alpha x$ can be easily obtained from $x$ by performing a bit shift. We denote the input $\boldsymbol{x} = x_1 x_2 \cdots x_n$, where $\boldsymbol{x} \in \{0,1\}^n$. Furthermore, there are $m$ output qubits denoted by $\ket{0}_1, \ket{0}_2, \dots , \ket{0}_m$. The value of $m$ is determined by the choice of $\alpha$.

    The quantum circuit is constructed in two parts. The first part resembles the circuit in Theorem~\ref{the:ReLU}: we start by applying an $X$ gate to $\ket{x_1}$, followed by $n-1$ Toffoli gates. The $i$-th Toffoli gate has $\ket{x_1}$ as the first control qubit, $\ket{x_{i+1}}$ as the second control qubit, and $\ket{0}_{i+1}$ as the target qubit. Finally, we apply another $X$ gate to $\ket{x_1}$. This first part of the circuit copies the state of $\ket{x_i}$ to $\ket{0}_i$ for $i=2,3,\dots, n$ when $x_1 = 0$, corresponding to the case where $x \geq 0$ and we should output $x$ itself.

    The second part of the circuit is constructed as follows: we apply $n-1$ Toffoli gates, where the $i$-th Toffoli gate has $\ket{x_1}$ as the first control qubit, $\ket{x_{i+1}}$ as the second control qubit, and $\ket{0}_{i+k}$ as the target qubit, with $k = 1 + \log \frac{1}{\alpha}$. Then, we apply a CNOT gate with $\ket{x_1}$ as the control qubit and $\ket{0}_1$ as the target qubit. If we use the True form, the circuit is complete. However, if we use the 2's complement, we need to apply $k-1$ additional CNOT gates, where the control qubit is $\ket{x_1}$ and the target qubits are $\ket{0}_2, \ket{0}_3, \dots, \ket{0}_{k}$ respectively. An example of a quantum circuit implementing the $5$-bit Leaky ReLU function with $\alpha = 0.125$ is depicted in Fig.~\ref{fig:lReLU}. In this figure, the circuit constructed in the first part is shown on the left side of the barrier. When using the True form, the gates inside the dashed box need to be omitted.

    To verify the correctness of the circuit, let's examine the value of $x_1$.

    When $x_1 = 0$, all quantum gates in the second part can be omitted. The circuit's correctness follows from the same proof as in Theorem~\ref{the:ReLU}.

    When $x_1 = 1$, all quantum gates in the first part can be omitted. The circuit will then copy $\ket{x_i}$ to $\ket{0}_{i+k}$ for $i=2,3,\dots,n$, where $k = 1 + \log \frac{1}{\alpha}$. Subsequently, we copy $\ket{x_1}$ to $\ket{0}_1$. If we use the True form, we already obtain $\alpha x$. However, if we use 2's complement, we need to flip $x_2,\dots,x_k$. Thus, we successfully store $\alpha x$ on the output qubits.

    To compute the $T$-depth of our quantum circuit, we observe that only the Toffoli gates contribute to it. By utilizing the method illustrated in Fig.~\ref{fig:4Toffoli} \cite{gokhale2021quantum}, we find that the $T$-depth of our circuit is 8. Similarly, we can determine that the circuit depth and size are both $O(\log n)$ and $O(n)$ respectively.
\end{proof}

\section{Implementation of activation function using Quantum Lookup Table}
\label{sec:QLUT}

The definition of a look-up table is a table that stores the output corresponding to all possible inputs. Rajiv Krishnakumar et al.~\cite{krishnakumar2022aq} proposed a method to implement a Quantum Look-up Table (QLUT) designed for quantum arithmetic functions. The QLUT offers a novel approach to data access and storage within quantum algorithms, facilitating the efficient retrieval and utilization of specific values or outcomes necessary for quantum arithmetic and other computational processes. This innovation is particularly advantageous for simplifying complex computational tasks.

To utilize QLUT for implementing activation functions, we initially determine the number of bits required to encode both the input and output of the function. For instance, if the input and output are encoded using $n$ and $m$ bits respectively, the function can be represented as $f(\boldsymbol{x}): \{0,1\}^n \to \{0,1\}^m$. Next, we generate a table containing all the possible inputs along with their corresponding outputs for the function. Finally, we employ this table to construct a quantum circuit $C$ such that:
\begin{equation*}
C\ket{\boldsymbol{x}}\ket{00\dots 0}\ket{00 \dots 0} = \ket{\boldsymbol{x}}\ket{f(\boldsymbol{x})}\ket{garbage_{\boldsymbol{x}}}.
\end{equation*}

Where $garbage_{\boldsymbol{x}}$ represents the garbage qubits. We utilize QLUT to implement various activation functions, including Sigmoid, Tanh, Swish, Softmax, ELU, and GELU. The input and output of these activation functions are floating-point numbers encoded using different numbers of bits: 8 bits, 16 bits, 32 bits, 64 bits, and 128 bits. We adhere to the IEEE 754 Standard \cite{kahan1996ieee} for encoding floating-point numbers, except for 8 bits. For 8-bit floating-point numbers, the encoding comprises one sign bit, four exponent bits, and three mantissa bits. In this study, both the input and output consistently employ the same encoding method. Therefore, the number of bits used to encode the input and output remains consistent.

To implement a specific activation function, such as the Sigmoid function, using QLUT, we first determine the number of bits to store the input and output, such as 32 bits. Subsequently, a table is generated containing $2^{32}$ entries, each entry representing the input and its corresponding output with 64 bits, fully describing the Sigmoid function. This table is then used to construct a quantum circuit.

Note that the choice of coding method only affects the table. One is free to use fixed-point numbers to generate the table and utilize our Qiskit code to construct the quantum circuit. Next, we introduce how to construct the circuit from the table.

\subsection{Quantum circuit implementation of quantum Lookup table}

We employ the SELECTSWAP network \cite{low2018trading} to implement the quantum Lookup Table with some minor revisions. The SELECTSWAP network comprises two parts: SELECT and SWAP. In the SELECT part, each term of the table requires a multi-controlled gate for implementation. The role of this gate is to prepare the output of each term to the output qubits only if the input state matches the input of that term. Naturally, the number of multi-controlled gates should match the number of terms in the table. If the input requires $n$ bits for storage, then there are $2^{n}$ multi-controlled gates, leading to a significant circuit depth. By leveraging the SWAP operator, we can mitigate the circuit depth while increasing the number of qubits. Let $l$ denote the number of swap qubits, where $0 < l < n$. For each multi-controlled gate, the $2^l$ outputs are prepared in {$n\cdot2^l$} qubits. Consequently, the number of multi-controlled gates can be reduced to $2^{n-l}$. After applying these gates, swap gates controlled by the swap qubits ensure that the correct output is directed to the output qubits. Further detailed explanation follows:

\textbf{SELECT:} The SELECT operator applies different sets of unitaries (here, the $X$ gate) based on the state of the control qubits. SELECT operator iterates through all possible addresses. Let's assume the $i$-th term of the table contains input $\boldsymbol{a}_i = a_{i,1}a_{i,2}\cdots a_{i,n} \in \{0,1\}^n$, and the corresponding output $\boldsymbol{b}_i = b_{i,1}b_{i,2}\cdots b_{i,n} \in \{0,1\}^n$. The SELECT operator is described as follows:
\begin{equation*}
    SELECT = \Sigma{i}\ket{\boldsymbol{a}_i}\bra{\boldsymbol{a}_i} \otimes X^{\boldsymbol{b_i}},
\end{equation*}
where $X^{\boldsymbol{b_i}} = \bigotimes_{j} X^{b_{i,j}}$. When implementing the $i$-th term of the table, two key points are considered in constructing the circuit: the control condition involves the input qubits in state $\ket{\boldsymbol{a}_i}$, and the $X$ gate is applied to the output qubits based on $\boldsymbol{b}_i$ and the control condition. To achieve this, we can utilize the quantum circuit of the quantum fan-out gate as depicted in Fig.~\ref{fig:F_8}, with the top CNOT gate replaced by an $n$-Toffoli gate and $X$ gates added on both sides of the $n$-Toffoli gate according to $\boldsymbol{a}_i$. Note that according to Fig.~\ref{fig:toffoli}, the $T$-depth of the Toffoli gate is $4$. Additionally, based on Lemma 7.2 of \cite{barenco1995elementary}, the $T$-depth of a $(k > 3)$-Toffoli gate is $16k - 32$.

\textbf{SWAP:} Assuming the number of swap qubits is $l$, where $0 < l < n$, there will be $2^l$ output registers. The function of the SWAP operator is to exchange the order of output registers so that the first register contains the quantum state of the register indexed by the state of the swap qubits. The SWAP operator is defined as follows:
\begin{equation*}
        SWAP[\ket{\boldsymbol{x}} \bigotimes_{i}\ket{\phi_{i}}_{i}] = \ket{\boldsymbol{x}}\ket{\phi_{\boldsymbol{x}}}_{0} \otimes \cdots .
\end{equation*}
The SWAP operator comprises controlled-SWAP gates \cite{low2018trading}. The crucial aspect is to handle the swap qubits one by one. An example of a quantum circuit with three swap qubits is depicted in Fig.~\ref{fig:SWAP}. It's worth noting that for many controlled-SWAP gates, the $T$-depth of the quantum circuit implementing them can be reduced to just $4$, given the condition that they share the same control qubit and have different target qubits for each SWAP gate. This reduction is achieved using the shared-control Toffoli technique~\cite{gokhale2021quantum} mentioned in Theorem \ref{the:ReLU}. An example of four controlled-SWAP gates is illustrated in Fig.~\ref{fig:4cswap}. Therefore, the $T$-depth of the SWAP operator is $4l$.

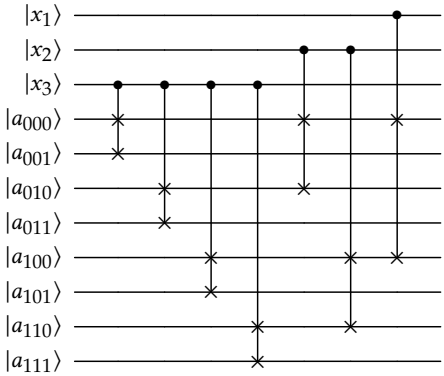
\begin{figure}[h]
\centering
\mbox{\Qcircuit @C=1.7em @R=1.2em @!R{
    \lstick{\ket{x_1}} & \qw & \qw & \qw & \qw & \qw & \qw & \ctrl{3} & \qw \\
    \lstick{\ket{x_2}} & \qw & \qw & \qw & \qw & \ctrl{2} & \ctrl{6} & \qw & \qw \\
    \lstick{\ket{x_3}} & \ctrl{1} & \ctrl{3} & \ctrl{5} & \ctrl{7} & \qw & \qw & \qw & \qw \\
    \lstick{\ket{a_{000}}} & \qswap \qwx[1] & \qw & \qw & \qw & \qswap \qwx[2] & \qw & \qswap \qwx[4] & \qw \\
    \lstick{\ket{a_{001}}} & \qswap & \qw & \qw & \qw & \qw & \qw & \qw & \qw \\
    \lstick{\ket{a_{010}}} & \qw & \qswap \qwx[1] & \qw & \qw & \qswap & \qw & \qw & \qw \\
    \lstick{\ket{a_{011}}} & \qw & \qswap & \qw & \qw & \qw & \qw & \qw & \qw \\
    \lstick{\ket{a_{100}}} & \qw & \qw & \qswap \qwx[1] & \qw & \qw & \qswap \qwx[2] & \qswap & \qw \\
    \lstick{\ket{a_{101}}} & \qw & \qw & \qswap & \qw & \qw & \qw & \qw & \qw \\
    \lstick{\ket{a_{110}}} & \qw & \qw & \qw & \qswap \qwx[1] & \qw & \qswap & \qw & \qw \\
    \lstick{\ket{a_{111}}} & \qw & \qw & \qw & \qswap & \qw & \qw & \qw & \qw \\
}}
\caption{The quantum circuit of SWAP operator for $\boldsymbol{x} \in \{0,1\}^3$.}
\label{fig:SWAP}
\end{figure}

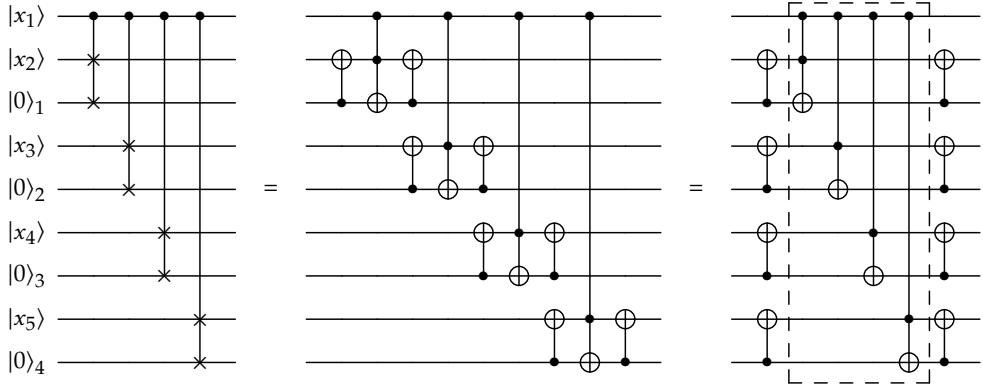
\begin{figure*}
\centering
\mbox{\Qcircuit @C=0.67em @R=1em @!R @!C{
    \lstick{\ket{x_1}} & \ctrl{1} & \ctrl{3} & \ctrl{5} & \ctrl{7} & \qw &&& \qw & \ctrl{1} & \qw & \ctrl{3} & \qw & \ctrl{5} & \qw & \ctrl{7} & \qw & \qw &&& \qw & \ctrl{1} & \ctrl{3} & \ctrl{5} & \ctrl{7} & \qw & \qw \\
    \lstick{\ket{x_2}} & \qswap \qwx[1] & \qw & \qw & \qw & \qw &&& \targ & \ctrl{1} & \targ & \qw & \qw & \qw & \qw & \qw & \qw & \qw &&& \targ & \ctrl{1} & \qw & \qw & \qw & \targ & \qw \\
    \lstick{\ket{0}_1} & \qswap & \qw & \qw & \qw & \qw &&& \ctrl{-1} & \targ & \ctrl{-1} & \qw & \qw & \qw & \qw & \qw & \qw & \qw &&& \ctrl{-1} & \targ & \qw & \qw & \qw & \ctrl{-1} & \qw \\
    \lstick{\ket{x_3}} & \qw & \qswap \qwx[1] & \qw & \qw & \qw &&& \qw & \qw & \targ & \ctrl{1} & \targ & \qw & \qw & \qw & \qw & \qw &&& \targ & \qw & \ctrl{1} & \qw & \qw & \targ & \qw \\
    \lstick{\ket{0}_2} & \qw & \qswap & \qw & \qw & \qw &\push{=}&& \qw & \qw & \ctrl{-1} & \targ & \ctrl{-1} & \qw & \qw & \qw & \qw & \qw &\push{=}&& \ctrl{-1} & \qw & \targ & \qw & \qw & \ctrl{-1} & \qw \\
    \lstick{\ket{x_4}} & \qw & \qw & \qswap \qwx[1] & \qw & \qw &&& \qw & \qw & \qw & \qw & \targ & \ctrl{1} & \targ & \qw & \qw & \qw &&& \targ & \qw & \qw & \ctrl{1} & \qw & \targ & \qw \\
    \lstick{\ket{0}_3} & \qw & \qw & \qswap & \qw & \qw &&& \qw & \qw & \qw & \qw & \ctrl{-1} & \targ & \ctrl{-1} & \qw & \qw & \qw &&& \ctrl{-1} & \qw & \qw & \targ & \qw & \ctrl{-1} & \qw \\
    \lstick{\ket{x_5}} & \qw & \qw & \qw & \qswap \qwx[1] & \qw &&& \qw & \qw & \qw & \qw & \qw & \qw & \targ & \ctrl{1} & \targ & \qw &&& \targ & \qw & \qw & \qw & \ctrl{1} & \targ & \qw \\
    \lstick{\ket{0}_4} & \qw & \qw & \qw & \qswap & \qw &&& \qw & \qw & \qw & \qw & \qw & \qw & \ctrl{-1} & \targ & \ctrl{-1} & \qw &&& \ctrl{-1} & \qw & \qw & \qw & \targ & \ctrl{-1} & \qw \gategroup{1}{22}{9}{25}{0.9em}{--}
}}
\caption{The parallel implementation of $4$ controlled swap gates.}
\label{fig:4cswap}
\end{figure*}

\begin{table*}
\caption{The $T$-depth of the quantum circuit implementing Sigmoid, SoftMax, Tanh, Swish, ELU, and GELU using QLUT with varying ancilla counts (\#ancilla) is shown. The number of input qubits ranges from $\{8,16,32,64,128\}$.}
\bgroup
\renewcommand{\arraystretch}{1.3}
\label{table:cost}
\centering
\resizebox{0.9\linewidth}{!}{
\begin{tabular}{c|c|c|c|c|c|c|c}
\hline \hline
\multicolumn{8}{c}{Number of input qubits = 8} \\ \hline
\#ancilla & 16 & 32 & 64 & 128 & 256 & 512 & 1024 \\ \hline
$T$-depth & 10244 & 4104 & 1548 & 528 & 148 & 40 & 28 \\ \hline \hline
\multicolumn{8}{c}{Number of input qubits = 16} \\ \hline
\#ancilla & 64 & 256 & 1024 & 4096 & 16384 & 65536 & 2.62 $\times 10^5$\\ \hline
$T$-depth & 3.15 $\times 10^6$ & 6.55 $\times 10^5$ & 1.31 $\times 10^5$ & 24608 & 4136 & 560 & 72\\ \hline \hline
\multicolumn{8}{c}{Number of input qubits = 32} \\ \hline
\#ancilla & 512 & 8192 & 1.31 $\times 10^5$ & 2.1 $\times 10^6$ & 3.36 $\times 10^7$ & 5.37 $\times 10^8$ & 8.59 $\times 10^9$ \\ \hline
$T$-depth &1.12 $\times 10^{11}$ & 5.91 $\times 10^9$ & 3.2 $\times 10^8$ & 1.47 $\times 10^7$ & 6.55 $\times 10^5$ & 24672 & 624 \\ \hline \hline
\multicolumn{8}{c}{Number of input qubits = 64} \\ \hline
\#ancilla & 16384 & 4.19 $\times 10^{6}$ & 1.07 $\times 10^{9}$ & 2.75 $\times 10^{11}$ & 7.04 $\times 10^{13}$ & 1.81 $\times 10^{16}$ & 4.61 $\times 10^{18}$ \\ \hline
$T$-depth & 6.23 $\times 10^{19}$ &2.07 $\times 10^{17}$ &6.69 $\times 10^{14}$ &2.06 $\times 10^{12}$ &5.91 $\times 10^{9}$ &1.47 $\times 10^{7}$ &24800\\ \hline \hline
\multicolumn{8}{c}{Number of input qubits = 128} \\ \hline
\#ancilla & 8.39 $\times 10^{6}$ & 5.5 $\times 10^{11}$ & 3.6 $\times 10^{16}$ & 2.36 $\times 10^{21}$ & 1.55 $\times 10^{26}$ & 1.01 $\times 10^{31}$ & 6.65 $\times 10^{35}$ \\ \hline
$T$-depth &9.14 $\times 10^{36}$ &1.19 $\times 10^{32}$ &1.51 $\times 10^{27}$ &1.83 $\times 10^{22}$ &2.07 $\times 10^{17}$ &2.06 $\times 10^{12}$ &1.47 $\times 10^{7}$ \\ \hline \hline

\end{tabular}
}
\egroup
\end{table*}

\textbf{SELECTSWAP:} The SELECTSWAP circuit is a hybrid of the SELECT and SWAP circuits. Assuming the number of input qubits and output qubits is $n$, and the number of swap qubits is $l$, the SELECT operator consists of $2^{n-l}$ steps. In each step, $2^l$ $n$-qubit registers are prepared to store all possible outputs based on all possible states of the $l$ swap qubits. In the SWAP operator, the register indexed by the quantum state of the swap qubits is placed in the first register. Consequently, the first register stores the desired output. The $T$-depth of circuit is $(2^{n-l})(16(n-l)-32) + 4l$ according to the discussion above. Note that the garbage qubits used in the SWAP operator could be replaced by dirty ancillary qubits \cite{low2018trading}.

\subsection{Analysis and Comparison of QLUT}

Note that the $T$-depth of the quantum circuit for QLUT is independent of the function it needs to implement; it is solely determined by the number of input qubits and swap qubits. Thus, the $T$-depth of the Sigmoid, Tanh, Swish, Softmax, ELU, and GELU functions is presented in Table~\ref{table:cost}. Here, we consider both the trade-off between qubit number and implementation accuracy, and between space and time. Since errors arise solely from coding accuracy, having more input qubits leads to a lower error rate. Increasing the number of swap qubits can significantly reduce the $T$-depth while also increasing the number of ancillary qubits. Users can choose the most suitable parameters based on their specific situation. For instance, a quantum circuit with a $T$-depth of 148 and 256 ancillary qubits could be utilized to implement activation functions such as Sigmoid and Tanh, where the input is an 8-bit floating-point number. The discussion regarding the opportunity cost of ancillary qubits explains that in this case, it is worthwhile to increase the number of ancillary qubits to reduce the $T$-depth \cite{gidney2018halving}.

\begin{table}[h]
\centering
\caption{The degree of the polynomial required to approximate the sigmoid function with the same error rate $\epsilon$ as using QLUT to approximate the sigmoid.}
\bgroup
\renewcommand{\arraystretch}{1.3}
\label{table:degree}
\centering
\resizebox{0.48\textwidth}{!}{
\begin{tabular}{c|c|c}
\hline
input qubit number & $\epsilon$ & degree of polynomial  \\ \hline \hline
$n=8$ & $1/2$ & 137 \\ \hline
$n=16$ & $1/2^{7}$  & 297 \\ \hline
$n=32$ & $1 / 2^{21}$ & 678 \\ \hline
$n=64$ & $1 / 2^{50}$ & 1469 \\ \hline
$n=128$ & $1 / 2^{112}$ & 3158 \\ \hline

\end{tabular}
}
\egroup
\end{table}

We illustrate the error rate of our activation function implementation and compare it with other methods. For simplicity, we define the error rate $\epsilon$ as follows, where $|x| < 15/4$:
\begin{equation*}
    |\hat{f}(x)-f(x)| < \epsilon.
\end{equation*}
Here, $\hat{f}(x)$ represents the function used to approximate $f(x)$. The error rate of using $8,16,32,64,$ and $128$-bit floating-point numbers to represent $|x| < 15/4$ is computed and shown in Table~\ref{table:degree}. We calculate the degree of polynomial needed to approximate the sigmoid function with the same error rate using Lemma 2 in \cite{livni2014computational}, as depicted in Table~\ref{table:degree}. The table demonstrates that even with a high error rate, the required degree remains relatively high, resulting in a high quantum cost for implementing sigmoid using polynomial approximation methods such as Quantum Singular Value Transformation (QSVT) \cite{gilyen2019quantum}.

\section{Conclusion}
\label{sec:concl}

This paper focuses on implementing activation functions with lower $T$-depth quantum circuits, yielding significant breakthroughs. 
Specifically, we achieve constant $T$-depth circuits for both ReLU and Leaky ReLU functions, regardless of the number of input qubits, $n$. Even when restricting the connectivity topology between qubits to a 2D grid, the implementation of ReLU function remains constant in $T$-depth. 
For other activation functions, such as Sigmoid, Tanh, ELU, Swish, SoftMax, and GELU, which are not conveniently implemented directly with arithmetic circuits, we employ quantum Lookup tables. 
We have balanced the number of input qubits against implementation precision and the trade-off between circuit $T$-depth and ancilla count to enhance the applicability of our results. Our method demonstrates advantages through a specific analysis of $T$-depth comparisons.

Looking ahead, we anticipate leveraging our activation function implementations in the implementation of quantum neural networks, thereby paving the way for future research directions in this domain.

\bibliographystyle{abbrv}
\bibliography{ref.bib}
\end{document}